\documentclass[letterpaper, 10 pt, conference]{ieeeconf}
\IEEEoverridecommandlockouts
\usepackage{cite}
\usepackage{amsmath,amssymb,amsfonts,amsthm}
\usepackage{graphicx}
\usepackage{cases}
\usepackage{xcolor, soul}
\usepackage{algorithm, algpseudocode}
\usepackage{url}

\usepackage[colorinlistoftodos]{todonotes}
\usepackage{textcomp}

\theoremstyle{plain}
\newtheorem{prop}{Proposition}

\theoremstyle{definition}

\newtheorem{ex}{Example}
\newtheorem*{rem*}{Remark}

\renewcommand{\rm}[1]{\mathrm{#1}}
\renewcommand{\bf}[1]{\mathbf{#1}}
\newcommand{\bb}[1]{\mathbb{#1}}
\newcommand{\cl}[1]{\mathcal{#1}}
\newcommand{\R}{\bb{R}}

\newcommand{\qmin}{q_c^{-}}
\newcommand{\qmax}{q_c^{+}}
\newcommand{\qmina}[1]{q_{#1}^{-}}
\newcommand{\qmaxa}[1]{q_{#1}^{+}}
\newcommand{\qin}{q^{\rm{in}}}
\newcommand{\qs}{q^{\rm{sh}}}
\newcommand{\C}{\cl{C}}
\renewcommand{\O}{\cl{O}}
\newcommand{\qtot}{q^{\rm{tot}}}
\newcommand{\qk}{q^{\rm{k}}}

\begin{document}
\title{Utilizing Load Shifting for Optimal Compressor Sequencing in Industrial Refrigeration}

\author{Rohit Konda, Vikas Chandan, Jesse Crossno, Blake Pollard, Dan Walsh, Rick Bohonek, and Jason R. Marden \thanks{R. Konda (\texttt{rkonda@ucsb.edu}) and J. R. Marden (\texttt{jrmarden@ece.ucsb.edu})are with the Department of Electrical and Computer Engineering at the University of California, Santa Barbara, CA. Vikas Chandan \texttt{vikas@crossnokaye.com}, Jesse Crossno \texttt{crossno@crossnokaye.com}, Blake Pollard \texttt{blake@crossnokaye.com}, and Dan Walsh \texttt{dan@crossnokaye.com} are with CrossnoKaye\textregistered . Rick Bohonek \texttt{rbohonek@butterball.com} is with Butterball LLC\textregistered. This work is supported by funding from CrossnoKaye.}}

\maketitle
\thispagestyle{empty}

\begin{abstract}
The ubiquity and energy needs of industrial refrigeration has prompted several research studies investigating various control opportunities for reducing energy demand. This work focuses on one such opportunity, termed \emph{compressor sequencing}, which entails intelligently selecting the operational state of the compressors to service the required refrigeration load with the least possible work. We first study the static compressor sequencing problem and observe that deriving the optimal compressor operational state is computationally challenging and can vary dramatically based on the refrigeration load. Thus we introduce load shifting in conjunction with compressor sequencing, which entails strategically precooling the facility to allow for more efficient compressor operation. Interestingly, we show that load shifting not only provides benefits in computing the optimal compressor operational state, but also can lead to significant energy savings. Our results are based on and compared to real-world sensor data from an operating industrial refrigeration site of Butterball LLC\textregistered located in Huntsville, AR, which demonstrated that without load shifting, even optimal compressor operation results in compressors often running at intermediate capacity levels, which can lead to inefficiencies. Through collected data, we demonstrate that a load shifting approach for compressor sequencing has the potential to reduce energy use of the compressors up to $20\%$ compared to optimal sequencing without load shifting.
\end{abstract}

\section{Introduction}
\label{sec:int}

Industrial refrigeration systems are present in a multitude of sectors, not limited to food processing, plastics, electronics, and chemical processing \cite{sun2022comprehensive, fabrega2010exergetic, stoecker1998industrial}. Altogether, industrial refrigeration accounts for approximately $8.4\%$ of total energy usage in the U.S \cite{eiareport}. As such, there are tremendous energy saving opportunities available in industrial refrigeration, not only through updating hardware components, but also increasing the sophistication of the implemented control algorithms. Algorithmic improvements are potentially more enticing, as they can realize significant energy savings with minimal capital expenditures to retrofit the system. 

The four central components of a prototypical refrigeration system include the evaporators, compressors, condensers, and the expansion valve, with interconnections as illustrated in Figure \ref{fig:ref}. We study a common industrial refrigeration process, where ammonia refrigerant is circulated in a closed loop in a vapor compression cycle to move heat against the thermal gradient of the system. Informally, the main thermodynamic steps in an ideal vapor compression cycle are summarized as follows\footnote{The actualized refrigeration process deviates from this idealization considerably, but we simplify for the purpose of presentation.}:

\vspace{5pt}
$(1 \to 2)$ The refrigerant vapor flows through a compressor, where it is compressed from a low pressure, referred to as suction pressure, to a high pressure, referred to as discharge pressure.  A consequence of this compression is an increase in the temperature of the refrigerant vapor, which now takes the form of a super-heated vapor.

$(2 \to 3)$ The refrigerant super-heated vapor is then fed to the condenser, where constant pressure heat rejection occurs, and heat is released to the ambient environment, resulting in condensation of the ammonia. A consequence of this heat rejection is that the refrigerant transitions from a super-heated vapor to a saturated liquid. 

$(3 \to 4)$ The refrigerant is then expanded adiabatically across an expansion valve, reducing the temperature and pressure and resulting in a vapor-liquid mixture.

$(4 \to 1)$ Cooled refrigerant liquid flows through the evaporator, where heat absorption from the system via evaporation of the refrigerant occurs, and super-heated vapor is fed back to the compressor, completing the cycle.

\vspace{5pt}

\noindent We direct the interested reader to \cite{stoecker1998industrial} for a comprehensive review of refrigeration systems.  

\begin{figure}[h!]
    \centering
    \includegraphics[width=150pt]{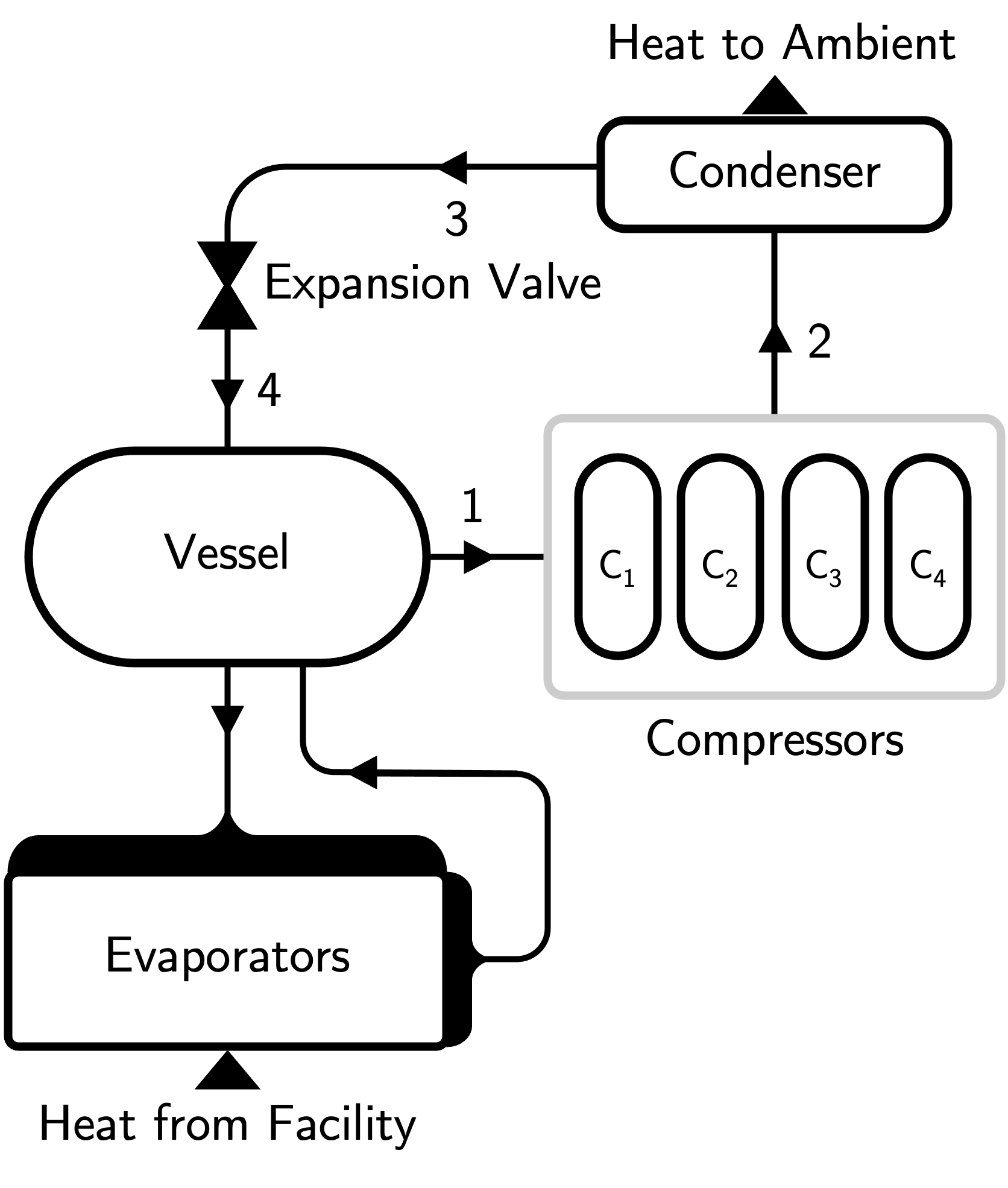}
    \caption{A simplified diagram of the refrigeration components are depicted showing the flow of ammonia through the vapor compression process.}
    \label{fig:ref}
\end{figure}

\begin{figure*}[ht]
    \centering
    \includegraphics[width=500pt]{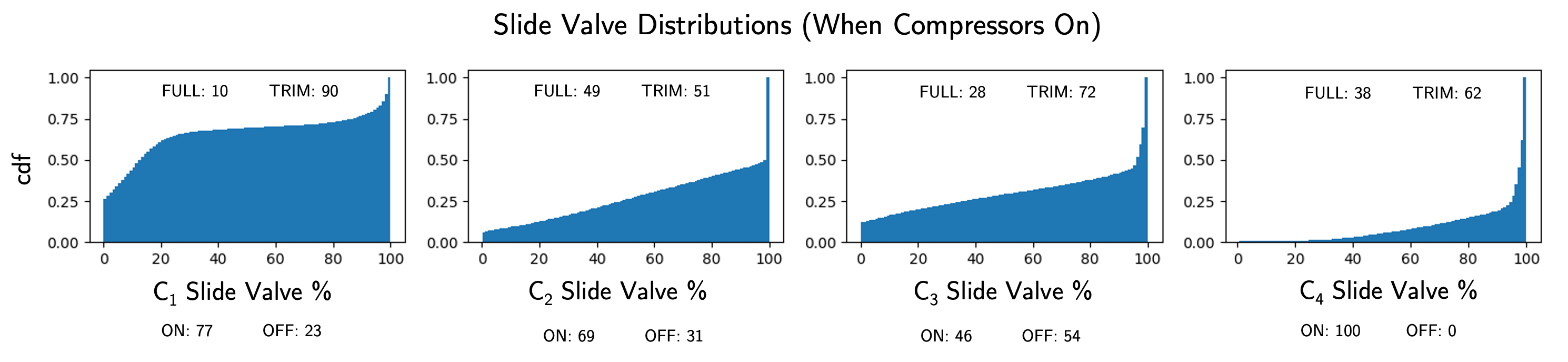}
    \caption{This figure highlights the cumulative distribution functions for the slide valve position for four compressors operating at the Butterball facility during the month of June, 2023. Here, the slide valve position is associated with compressor capacity, where $100\%$ means that the compressor is running at full capacity. We also highlight the percentages in which each compressor is operating at full capacity (where the slide valve sensor is measured above $99\%$) or trim as well as the percentage of time the compressor is turned on and off. Note that the compressors are often operating in trim, suggesting that there are potential opportunities to save energy by operating the compressors at full capacity more often.}
    \label{fig:sv1}
\end{figure*}

The configuration and operation of the entire refrigeration system can have significant impacts on the cost of operation; this can either be measured through total power, electric cost, carbon emissions, etc. Infrastructural retrofits of the refrigeration system, including changing the choice of refrigerant, hardware specifications of components, or general system layout, can be typically costly to implement. Accordingly, a more viable way to reduce costs is to strategically adjust the control policies of the refrigeration components to meet the required heat extraction while minimizing the operational cost. For example, \emph{thermal load shifting} has received significant attention as a methodology to preemptively cool a facility in order to take financial advantage of dynamic energy cost-rate structures \cite{sun2013peak, braun2003load, afram2014theory, yao2021state, pattison2016optimal, pattison2017moving, vishwanath2019iot}. Additionally, another approach is \emph{set point optimization}, where the set points for suction and discharge pressure are dynamically adjusted to drive the refrigeration system towards an energy optimal operating state while maintaining the desired achievable cooling demands \cite{larsen2003control, manske2000performance, larsen2006model}. For these domains, standard control and optimization techniques, such as model predictive control and dynamic programming, can be implemented for attaining near-optimal control strategies.

While most of the existing control approaches for energy optimization focus on the evaporators, e.g., thermal load shifting, it is important to highlight that the compressors represent the dominant energy expenditure in most refrigeration systems, typically accounting for more than $85\%$ of the total energy usage of the refrigeration process. For example, Figure \ref{fig:breakdown} highlights the approximate breakdown of the energy demand associated with an industrial refrigeration site of Butterball, which is large poultry processing facility, during the month of June, 2023. It is widely known that compressors are operated most efficiently when running at full capacity \cite{reindl2013sequencing, manske2001evaporative}; however, the typical control objective for the compressors is suction pressure stabilization. In this way, the operational state of the compressors is directly dependent on the state of the evaporators and this can ultimately lead to the compressors operating in an inefficient manner, i.e., at partial capacity. Figure~\ref{fig:sv1} confirms this phenomena, directly highlighting the cumulative distribution functions for the slide valve position of the four compressors operating at the Butterball facility during this time period.\footnote{This data was acquired through direct partnership with CrossnoKaye (see \url{crossnokaye.com}), which focuses on the derivation and implementation of intelligent control systems for industrial refrigeration systems in the cold food and beverage domain.  CrossnoKaye has been monitoring and controlling the refrigeration system at the Butterball facility since April 2023}  Here, the slide valve position can be viewed in the same light as capacity, where $100\%$ means that the compressor is running at full capacity. Note that the slide valve positions are often significantly below $100\%$, suggesting that there are potential energy saving opportunities in algorithmic improvements for compressor scheduling and control.

\begin{figure}[ht]
    \centering
    \includegraphics[width=250pt]{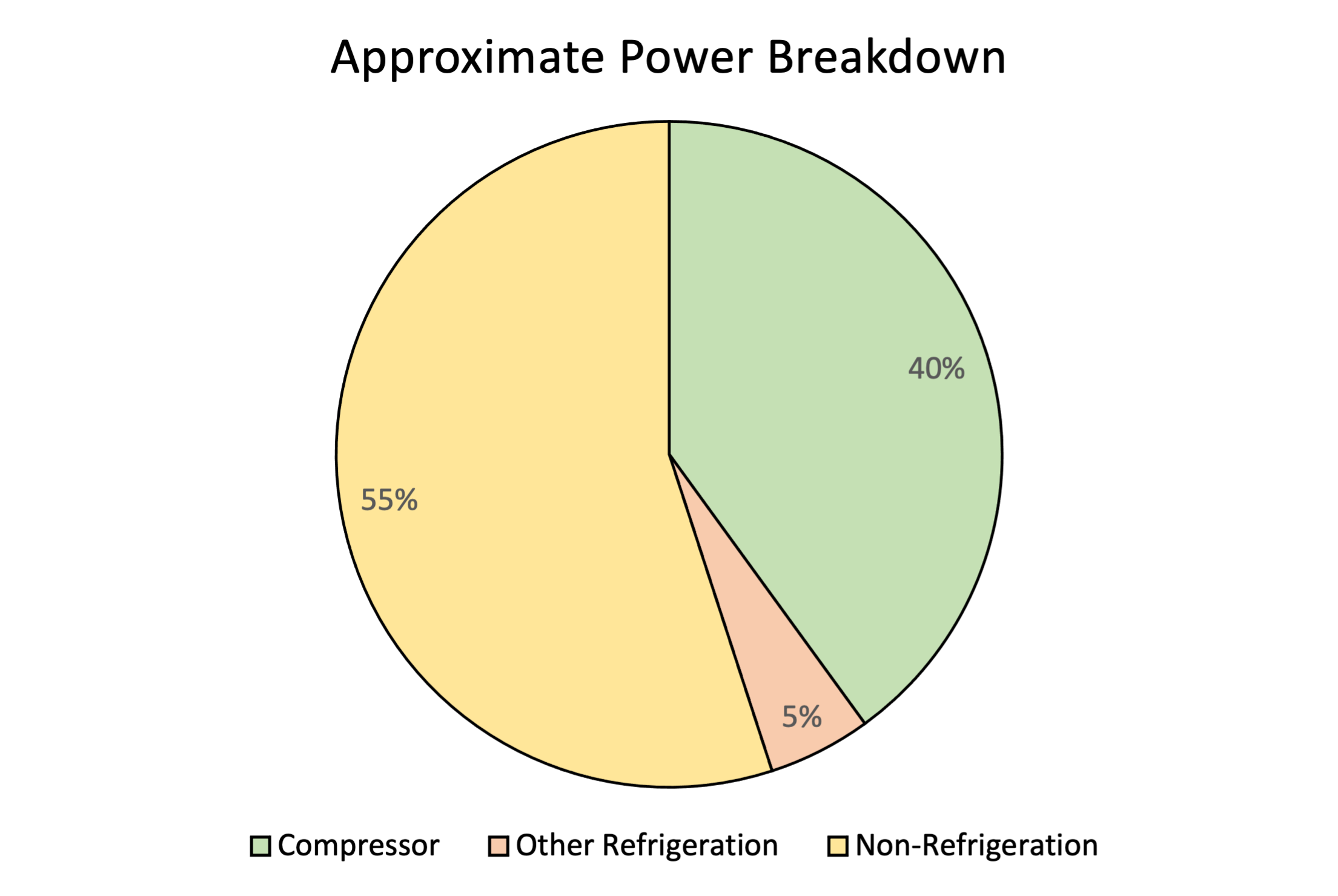}
    \caption{Approximate power breakdown for a Butterball facility at Huntsville.}
    \label{fig:breakdown}`
\end{figure}

This paper shifts the control focus from the evaporators directly to the compressors, where the goal is to optimize the operational state of the compressors to serve the required refrigeration load. We formalize this optimization problem as the \emph{compressor sequencing problem} \cite{reindl2013sequencing}. We begin by characterizing the optimal solution to the \emph{static} compressor sequencing problem, which focuses on satisfying a given refrigeration load at a single time instance. Interestingly, we show that the optimal solution can be derived via a simple water-filling algorithm. However, this water-filling algorithm must be executed with the correct compressor sequence, which is dependent on the refrigeration load and can be difficult to compute, to recover the optimal solution. We note that the optimal compressor sequence may also vary drastically as a function of the refrigeration load, hence creating significant issues from both an implementation and computation perspective.

Given these limitations, our second result shifts attention from the static compressor sequencing problem to the \emph{dynamic} compressor sequencing problem. Here, we are provided with a given time-dependent profile of the refrigeration load that we need to serve over a given horizon. See forthcoming Figure \ref{fig:heat} as an illustration of a typical refrigeration load over a month long horizon at the Butterball facility.  Unlike the static problem, this dynamic formulation gives us the flexibility to exploit load shifting, where one preemptively cools the facility (i.e., services future loads), so that the compressors can operate in a more efficient fashion, i.e., more often at full capacity. Our main results characterize the optimal solution of this dynamic compressor sequencing problem.  Not only does load shifting provide substantial potential for energy savings, interestingly, the resulting optimal solution is more implementable and tractable than the optimal solution of the static problem. In particular, the optimal solution is the result of water-filling algorithm at each instance, where now the order of the compressors is fixed and can be set a priori.  This ensures that the resulting behavior will be far more predictable and stable than just implementing the solution of the optimal static compressor sequencing each period, thus making it more amenable to real world implementation.

In order to practically realize the energy saving opportunities associated with compressor sequencing, we implement a numerical case study on the operation of the  Butterball facility in Section \ref{sec:case}. The basis of this numerical study was from collected time-series data on compressor configurations and refrigeration load estimates. Using these refrigeration load estimates, our initial results suggest that the potential energy savings could be significant, with upwards of $20\%$ reduction in total energy expenditure when comparing optimal compressor sequencing with load shifting to optimal compressor sequencing without load shifting. Furthermore, this paper provides a number of supporting results characterizing properties of the optimal and near-optimal online load shifting algorithms.  

While we introduce load shifting as a novel mechanism to address the compressor sequencing problem, practical implementations of the proposed algorithms would require additional practical considerations. Increasing or decreasing the capacities of the compressors carelessly can result in undesirable swings in suction and discharge pressure and lead to system instability. Classically, evaporators handle thermal regulation while the compressors handle pressure regulation. One possible way to implement load shifting is to directly control compressor operation while having the evaporators handle pressure regulation, either by turning additional evaporators on/off or adjusting the fan speed. Nevertheless, this work is merely gauged at assessing the potential benefits of  load shifting for compressor sequencing with the realization of practical control strategies to be relegated to future work.

\section{Optimal Compressor Sequencing}
\label{sec:results}

In this section, we formalize the control problem for optimal compressor sequencing. Here, the operational state of the compressors (e.g. the on/off status as well as the slide valve position) is chosen such that the thermal demands are met with the least cost, which we measure in terms of energy usage. We discuss potential opportunities for algorithmic improvements in this section.

\subsection{Preliminaries}

Many large scale refrigeration systems employ algorithms for intelligently choosing the operational state of the compressors. In refrigeration systems with multiple compressors, one must decide the operational state of these compressors that is necessary to service the underlying refrigeration load. More formally, let $\C$ denote a finite set of compressors (for Butterball, $\C = \{C_1, C_2, C_3, C_4\}$), where each compressor $c \in \C$ is associated with a minimum and maximum heat capacity, $\qmin$ and $\qmax$ respectively, as well as a power-heat curve $P_c: Q_c \to \R_{\geq 0}$ where $Q_c = 0 \cup [\qmin, \qmax]$ designates the viable refrigeration loads on compressor $c$, with $0$ indicating the compressor is turned off. Here, $P_c(q_c) \geq 0$ is the power required to serve heating load $q_c \in Q_c$ through compressor $c$. We assume that $P_c(0) = 0$ and $P_c$ is concave and increasing over the interval $[\qmin, \qmax]$, which implies that compressors operate more efficiently at higher capacities. Specifically for the compressors in operation at Butterball,  we assume an affine structure for the power-heat curves, where for any compressor $c \in {\cal C}$ and thermal load $q_c \in [\qmin, \qmax]$ we have

\begin{equation*}
P_c(q_c) = P_c(\qmin) + \left(\frac{q_c - \qmin}{\qmax-\qmin} \right) \left(P_c(\qmax)-P_c(\qmin)\right).
\end{equation*}

We validate the affine models for the power-heat curves against collected data on estimated refrigeration load and compressor power, which is shown in Figure \ref{fig:PH}. The extreme points of these power-heat curves is summarized in Table~\ref{tab:comp}.
\begin{table}[h!]
\centering
 \begin{tabular}{||c c c c c||} 
 \hline
  Compressor & C1 & C2 & C3 & C4 \\ [0.5ex] 
 \hline\hline
 Model & Screw & Screw & Screw & Screw \\ [1ex] 
 $\qmin$ (kW) & 220 & 239 & 165 & 284 \\ 
 $\qmax$ (kW) & 3000 & 2126 & 1760 & 2351 \\
 $P(\qmin)$ (kW) & 124 & 173 & 142 & 181 \\
 $P(\qmax)$ (kW) & 262 & 427 & 356 & 494 \\
 \hline
 \end{tabular}
 \caption{Compressor Characteristics}
 \label{tab:comp}
\end{table}
Note that for simplicity, we removed the dependence on slide valve position to provide a direct relationship between thermal load and power.

\begin{figure*}[t!]
    \centering
    \includegraphics[width=500pt]{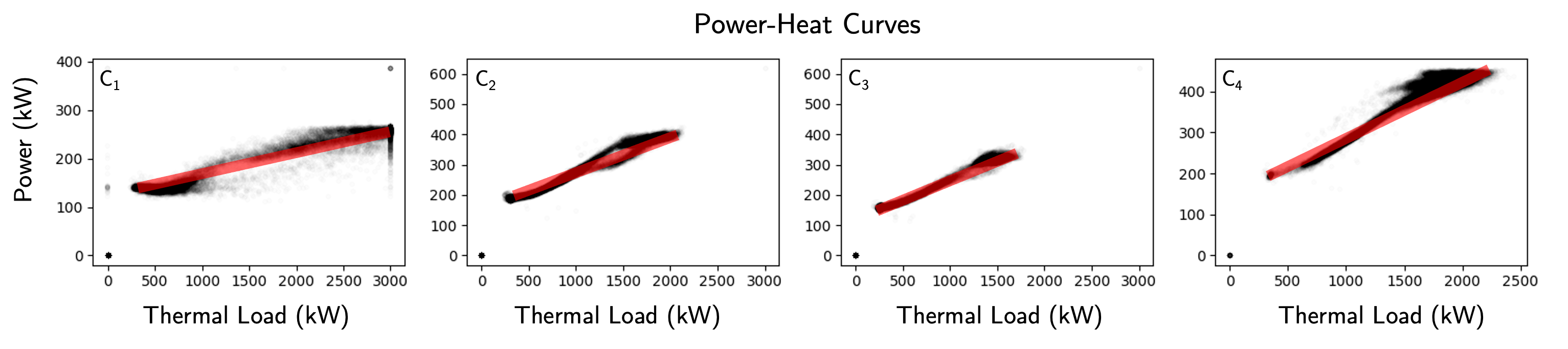}
    \caption{For each of the compressors, the estimated power and heat capacity for each minute in the month of June was recorded for a Butterball facility. We depict the resulting spread in the given figure and notice a fairly affine relationship, which we denote in red. This is also supported from manufacturing simulation software for the compressors.}
    \label{fig:PH}
\end{figure*}

The problem of compressor sequencing centers on the goal of meeting the incoming refrigeration load, which we denote by $\qin \in \R_{\geq 0}$, with the least possible energy expenditure.  More formally,  the goal is to identify compressor loads $\{q_c\}_{c \in \C}$ that satisfy the incoming refrigeration load, i.e., $\sum_{c \in \C} q_c \geq \qin$, and minimize the total work expenditure as measured by the total power usage by the compressors, i.e., $\sum_{c \in {\cal C}} P_c(q_c)$.  We denote this compressor assignment by the policy $\pi : \R_{\geq 0} \rightarrow \prod_{c \in {\cal C}} Q_c$, where $\pi(\qin) = \{q_c\}_{c \in {\cal C}}$ designates the refrigeration loads for each compressor $c \in \C$. The cost of a policy $\pi$ for a given $\qin$ is defined by $J_{\pi}(\qin) = \sum_{c \in \C} P_c(\pi_c(\qin))$. We will henceforth remove the dependence on the compressor set, i.e., denote $\{\cdot\}_{c \in {\cal C}}$ as merely $\{\cdot\}$, for notational simplicity.

\subsection{Fixed Order Compressor Sequencing}
\label{subsec:fixed}

The industrial standard for compressor operation is to meet a given refrigeration load $\qin$ through a water filling algorithm with a pre-determined order of  compressors $\C$. For ease of presentation, let the set of compressors $\C = \{c_1, c_2, \dots, c_m\}$ naturally denote the order of the compressors, i.e., $c_1$ first, $c_2$ second, etc. Then the operation of the compressors according to this policy, represented by $\pi^{\rm{FO}}$, is given by Algorithm \ref{alg:water}. 
\begin{algorithm}
\caption{Water Filling Algorithm}\label{alg:water}
\begin{algorithmic}
\Require $\O$, $\qin$, $\qtot \gets 0$, $q_c \gets 0$ for all $c \in \C$
\For{$c$ in $\O$}
\If{$\qin > \qtot$}
\State $q_c \gets \qmax$
\State $\qtot \gets \qtot + \qmax$
\EndIf
\EndFor
\For{$c$ in $\rm{reverse}(\O)$}
\If{$\qin \leq \qtot$ and $q_c \neq 0$}
\State $d \gets \min\{\qmax - \qmin, \qtot - \qin\}$
\State $q_c \gets q_c - d$
\State $\qtot \gets \qtot - d$
\EndIf
\EndFor
\State \Return $\{q_c\}_{c \in \C}$
\end{algorithmic}
\end{algorithm}
We will denote the fixed order policy as $\pi^{\rm{FO}}(\qin; \O)$, where $\O$ describes a specific ordering of the compressor set $\C$. Note that Algorithm \ref{alg:water} returns a thermal load profile $\{q_c\}$ that is guaranteed to satisfy the inequality $\sum_{c \in \C} q_c \geq \qin$ provided that we assume that $\min_c \qmin \leq \qin \leq \sum_{c \in \C} \qmax$. When $q_c = \qmax$, we say that the compressor is operating at full capacity. Alternatively, when $\qmax > q_c \geq \qmin$, we say that the compressor is operating in trim. A numerical example of Algorithm \ref{alg:water} is shown below.

\begin{ex}
Consider the set of compressors $\C$ associated with Butterball with order and characteristics following from Table \ref{tab:comp}. Let Algorithm \ref{alg:water} be run for an incoming refrigeration load $\qin = 3100$. Since $\qin > \qmaxa{1} = 3000$, we first set $q_1 \gets \qmaxa{1}$ and $q_2 \gets \qmaxa{2} = 2126$ according to the first loop in Algorithm \ref{alg:water}. Since $3000 + 2126 > 3100$, we run the second loop to scale back the compressor loads. First we set $q_2 \gets \qmina{2} = 239$, resulting in a total load of $3239 > 3100$. Lastly, to match $\qin$, we scale back $q_1 \gets 3000 - 139$ to result in compressor loads $q_1 = 2861$, $q_2 = 239$, $q_3 = 0$, and $q_4 = 0$.
\end{ex}

The central tuning parameter of the policy $\pi^{\rm{FO}}(\qin; \O)$ is the order $\O$ that the water filling algorithm is run under. There may be significant differences in energy usage for different orders, especially if compressors vary in the efficiency.  This phenomenon is displayed in Figure \ref{fig:diff}, where we see huge gap in energy usage between the best and worst compressor sequence when considering the four compressors in operation at Butterball.

\begin{figure}[ht]
    \centering
    \includegraphics[width=240pt]{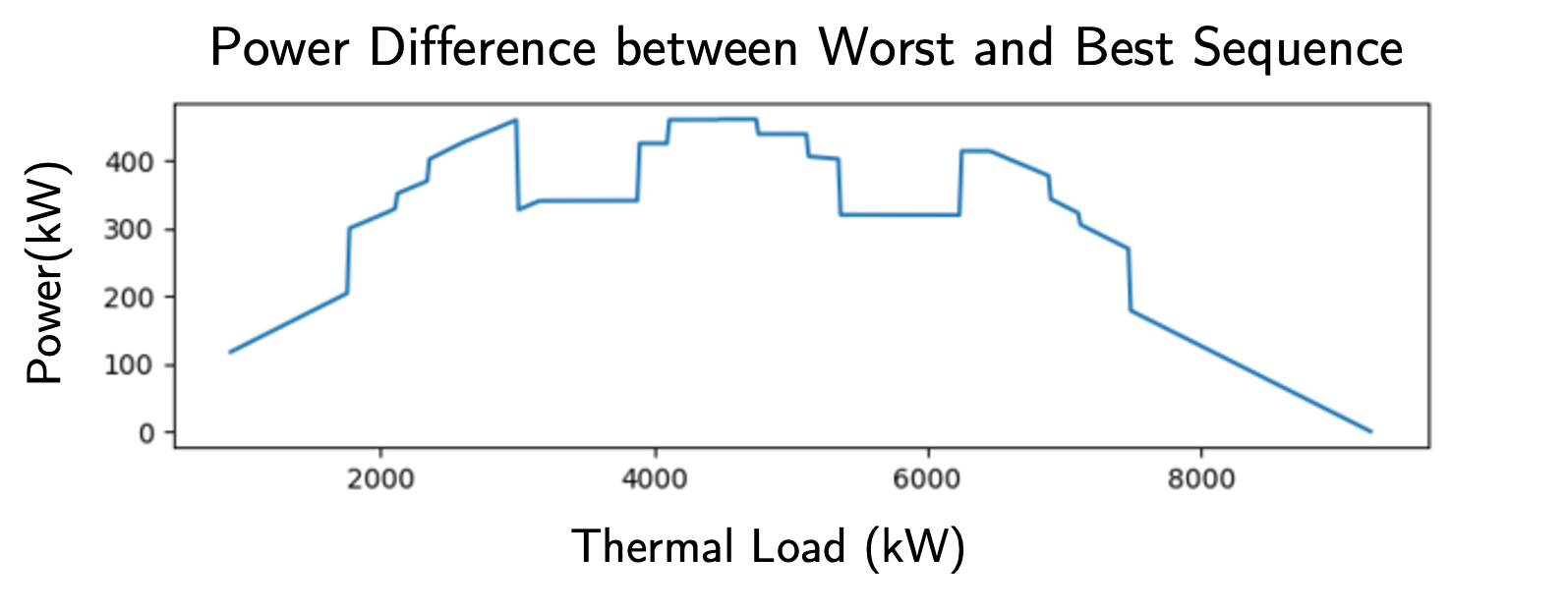}
    \caption{In this figure, for each possible refrigeration load, we calculate the power usage between water filling using the best and worst order of the four compressors in operation at Butterball. Notice that there can be potentially a large gap between power usage between the different orders.}
    \label{fig:diff}
\end{figure}

Note that for intermediate thermal loads, the efficiency loss according to a fixed-order water filling algorithm could be in upwards of $49\%$ (see Table \ref{table:main}).

\subsection{Optimal Compressor Sequencing}
\label{subsec:cs}

While the standard water-filling algorithms provide a straightforward approach to compressor sequencing, moving away from a fixed order scheme may lead to more efficient compressor operation. Hence, we consider the problem of optimal compressor sequencing in this section, where the goal is to determine the compressor state that meets the refrigeration load with the least possible energy expenditure. More formally, the operation of the compressors would be determined by the solution of the following non-convex optimization problem.
\begin{equation}
\label{opt-cs}
\begin{aligned}
J^*(\qin) &= \min_{q_c \in Q_c} \sum_{c \in \C} P_c(q_c) \\
&\text{s.t. }   \sum_{c \in \C} q_c \geq \qin
\end{aligned}    
\end{equation}

The following proposition characterizes the structure of optimal solution to the compressor sequencing problem in Eq. \eqref{opt-cs}. In fact, regardless of the refrigeration load $\qin$, the optimal compressor state can be realized by a water-filling algorithm with a specific compressor order that depends on $\qin$. 

\begin{prop}
\label{prop:noload}
Let $\qin$ be the incoming refrigeration load. The optimal compressor state, as given by the solution of Eq. \eqref{opt-cs}, can be realized by the water filling algorithm given in Algorithm \ref{alg:water} with a specific order $\O$ that depends on $\qin$.
\end{prop}
\begin{proof}
To show this statement, we show equivalently that for all the compressors, only one compressor $c \in \C$ has $\qmax > q^*_c > \qmin$ in the optimal solution to Eq. \eqref{opt-cs}. Note that if this is true, the optimal order $\O$ coincides to when the compressors are ordered decreasing in their capacities $\{q^*_c\}$. Given this order, Algorithm \ref{alg:water} will produce the equivalent capacity $\{q^{\rm{alg}}_c\} = \{q^*_c\}$ that match the optimal solution.

We show the claim that at most one $q^*_c \notin \{\qmin, \qmax, 0\}$ is not at the endpoints through contradiction. Let $i$ and $j$ be the compressors at partial capacity in the optimal solution. Notice that the function $P_i(q_i + d) + P_j(q_j - d)$ is a concave function of $d$, by assumption of concavity of $P_i$ and $P_j$ and preservation of concavity under affine transformations. For any feasible $d$, note that the constraint in Eq. \eqref{opt-cs} is always satisfied if we change $q_i$ to $q_i + d$ and $q_j$ to $q_j - d$. Additionally, since $P_i(q_i + d) + P_j(q_j - d)$ is concave, the optimal value for $d$ must occur at either endpoints of the feasible interval. Thus in the optimal solution, $q_i$ or $q_j$ must be either at $\qmin$ or $\qmax$ or $0$, ensuing in contradiction.
\end{proof}

This structural result demonstrates that the optimal control algorithm for compressor sequencing could be realized from the perspective of a partitioning process where the admissible refrigeration loads are partitioned into various regions, and each region is associated with a distinct ordering as shown in Figure \ref{fig:order}. However, this approach to compressor sequencing has significant problems from both a computation and implementation perspective. First, solving the optimization problem in Eq. \eqref{opt-cs} represents a mixed integer optimization that grows exponentially in complexity in the size of the compressor set $\C$. Furthermore, observe that the optimal order can change drastically as a function of the refrigeration load as highlighted in Figure \ref{fig:order}. This means that the state of the compressors could shift wildly during operation, which may be infeasible due to delays in changing compressor capacities and causing unnecessary variability in compressor operation. Thus, we look to load shifting as a medium to smooth out the compressor sequencing problem - we discuss this in the next section.

\begin{figure}[ht]
    \centering
    \includegraphics[width=240pt]{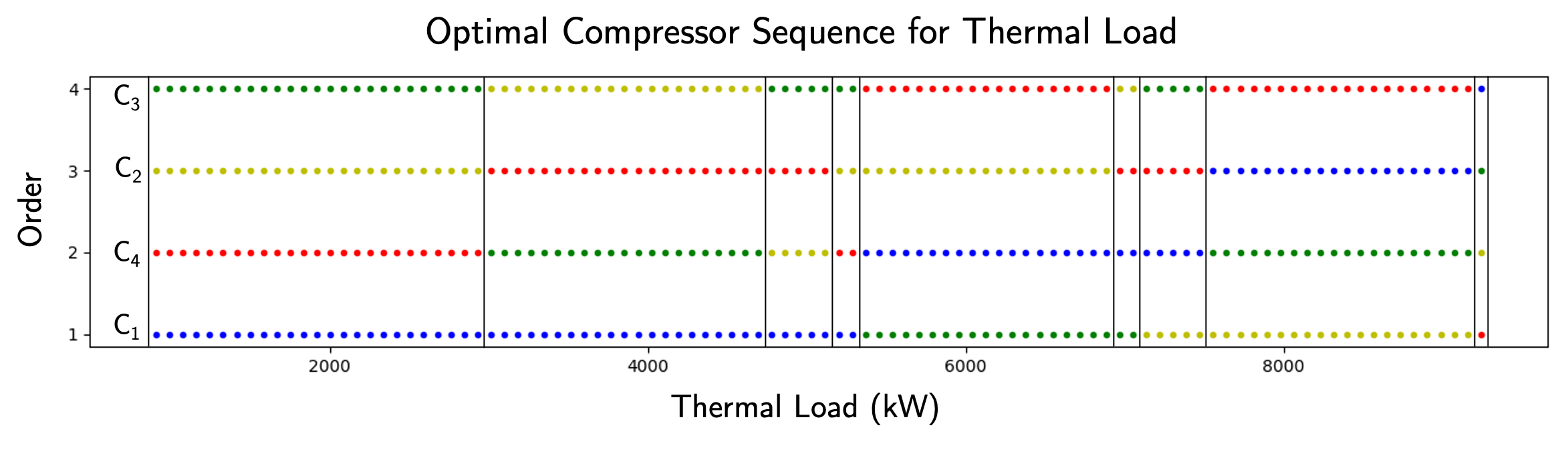}
    \caption{We depict the optimal compressor sequence for the four compressors in Butterball. For each given refrigeration load $\qin$, we pictorially describe the optimal order where each compressor is associated with a unique color and the height determines its placement in the order.}
    \label{fig:order}
\end{figure}

\subsection{Optimal Compressor Sequencing with Load Shifting}

Load shifting is common practice in refrigeration systems for reducing operational costs. Load shifting involves the process of preemptively cooling a facility, thereby using the product within the facility as a thermal battery to save on future cooling demands. Accordingly, for this setting we will think about cooling needs over a given discrete horizon $[0, 1, \dots, T]$ where the refrigeration load at each stage $k$ is given by $\qin(k)$ and $\bf{\qin} = \{\qin(k)\}_{0 \leq k \leq T}$. Here, we will assume that there is complete knowledge of the refrigeration load over the horizon at the initial stage $k=0$. This assumption will allow us to hypothetically assess the potential opportunities associated with load shifting for compressor sequencing on realistic refrigeration load profiles as provided in forthcoming Figure \ref{fig:heat}, which highlights the refrigeration load over the month of June, 2023.

The goal of optimal compressor sequencing is to establish a new shifted thermal load demand trajectory $\bf{\qs} = \{\qs(0), \dots, \qs(T)\}$ and dynamic compressor states $\bf{q}_c = \{q_c(0), \dots, q_c(T)\}\}$ that minimize the cumulative energy expenditure. Here, we require the shifted load demand trajectory satisfies
$$ \sum_{k=0}^\tau \qs(k) \geq \sum_{k=0}^\tau \qin(k), \ \forall \tau \in [0,T],$$
where the provided cooling exceeds the refrigeration load required for any horizon $[0,\tau]$ with $\tau \in \{0, \dots, T\}.$  Accordingly, our new optimization takes on the following form:
\begin{equation}
\label{opt-dcs}
\begin{aligned}
J^*(\bf{\qin})  &= \min_{\bf{\qs}, \{\bf{q}_c\}} \ \ \frac{1}{T} \sum_{k=0}^{T} \sum_{c \in \C} P_c(q_c(k)) \\
 \text{s.t. }  q_c(k) &\in Q_c \ \text{ for all } c \in \C, k \in [0,T],  \\
  \sum_{k=0}^\tau \qs(k) &\geq \sum_{k=0}^\tau \qin(k) \ \text{ for all } \tau \in [0,T],  \\
   \sum_{c \in \C} q_c(k) &\geq \qs(k) \ \text{ for all } k \in [0,T].
\end{aligned}
\end{equation}

This optimization has two sets of decision variables: the shifted thermal load trajectory $\bf{\qs}$ and the dynamic compressor loads $\{\bf{q}_c\}$. The first constraint ensures that the compressors loads are viable for every stage $k$. The second constraint dictates that that the shifted load trajectory $\bf{\qs}$ needs to deliver at least as much cooling as any nominal thermal profile $\bf{\qin}$ for any horizon $\tau \in [0,T]$. The last constraint ensures that the total compressor load needs to match the shifted thermal load $\bf{\qs}$ at each stage $k$. The optimization problem in Eq. \eqref{opt-dcs} is determined through $\{\bf{q}_c\}$ and $\bf{\qs}$, whereas the static compressor sequencing problem in Eq. \eqref{opt-cs} fixes $\qs(k) = \qin(k)$ for all $k$. While this results in a seemingly more complex optimization problem, the following proposition demonstrates that the optimal solution is actually attained by a fixed order water-filling algorithm.

\begin{prop}
\label{prop:load}
Let $\bf{\qin}$ be the dynamic refrigeration load. Furthermore, consider a set of compressors $\C = \{c_1, \dots, c_m\}$ that are ordered in terms of increasing marginal costs of cooling at full capacity, i.e., if compressor $i$ comes before compressor $j$ in the order $\O^{\rm{sh}}$, we have that
\begin{equation}
\frac{\qmaxa{i}}{P_{i}(\qmaxa{i})} \leq \frac{\qmaxa{j}}{P_{j}(\qmaxa{j})}.
\end{equation}
Then the following optimization problem yields the same optimal cost in Eq. \eqref{opt-dcs} when $T \to \infty$:
\begin{equation}
\label{opt-dcs2}
\begin{aligned}
J^*(\bf{\qin}) &= \min_{\bf{\qs}} \ \ \frac{1}{T} \sum_{k=0}^{T} \sum_{c \in \C} P_c(q_c(k)) \\
\text{ s.t. } q_c(k) &= \pi_c^{\rm{FO}}(\qs(k)) \text{ for all } c \in \C, k \in [0,T], \\
\sum_{k=0}^\tau \qs(k) &\geq \sum_{k=0}^\tau \qin(k), \ \text{ for all } \tau \in [0,T]  
\end{aligned} 
\end{equation}
where $\pi_c^{\rm{FO}}(\cdot)$ comes from the water filling algorithm in Algorithm \ref{alg:water} with the above ordering $\O^{\rm{sh}}$.  
\end{prop}
\begin{proof}
Let $\bf{\hat{\qs}}, \{\bf{\hat{q}}_c\}$ be the optimal solution for Eq. \eqref{opt-dcs}. As $P_c$ is monotonic, we note that $\sum_{c \in \C} \hat{q}_c(k) = \hat{\qs}(k)$ must be hold with equality for all $k$ in the optimal solution. Then, for a given $\hat{\qs}(k)$, the optimal compressor loads $\{\hat{q}_c(k)\}$ can be given through the water filling algorithm in Algorithm \ref{alg:water} for each $k$ for some order $\O(k)$. This can be shown with arguments similar to the proof of Proposition \ref{prop:noload}.

Now we show that $\O(k) = \O^{\rm{sh}}$ for all $k$. Let $i$ and $j$ be compressors such that $i$ comes before $j$ in $\O^{\rm{sh}}$ but $j$ comes before $i$ in $\O(k)$. We claim that this is only possible a finite number of times. If not, there are an infinite number of times where $\hat{q}_i(k) = 0$, but $\hat{q}_j(k) > 0$. However, there exist a load shift and time points $\{k_{\ell}\}_{1 \leq \ell \leq N}$ time steps in which $q_i(k_\ell) \to q \leq q_i^+$ for $1 \leq \ell \leq M < N$ and $q_j(k_\ell) \to 0$ for $M < \ell \leq N$ which produces a more efficient solution, since the marginal cost of cooling for compressor $i$ is less than compressor $j$. Thus the water filling algorithm with order $\O^{\rm{sh}}$ recovers an optimal solution to Eq. \eqref{opt-dcs} when $T \to \infty$.
\end{proof}

From Proposition \ref{prop:load}, we see that a simple fixed order water-filling algorithm achieves the optimal solution to Eq. \eqref{opt-dcs} given the correct shifted load $\bf{\qs}$. We can also evaluate the potential cost benefits between compressor sequencing with and without load shifting. We characterize the greatest possible difference in cost in the next proposition. For notational ease, we define the ratios $R_{\max} = \max_c P_c(\qmin)/\qmin$ and $R_{\min} = \min_c P_c(\qmax)/\qmax$.

\begin{prop}
\label{prop:comp}
For a given dynamic refrigeration load $\bf{\qin}$, let $J^{\rm{cs}}(\bf{\qin})$ be the trajectory cost in Eq. \eqref{opt-dcs} associated with compressor sequencing without load shifting and let $J^*(\bf{\qin})$ be the optimal trajectory cost with load shifting with $T \to \infty$. The fractional difference between the costs is upper and lower bounded by
\begin{equation}
0 \leq \frac{J^{\rm{cs}}(\bf{\qin}) - J^*(\bf{\qin})}{J^*(\bf{\qin})} \leq \frac{R_{\max} - R_{\min}}{R_{\min}}
\end{equation}
\end{prop}
\begin{proof}
We first note that in Eq. \eqref{opt-dcs}, we recover the optimization without load shifting (or Eq. \eqref{opt-cs}) if we impose the constraint $\bf{\qs} = \bf{\qin}$ directly. Therefore, we have that $J^*(\bf{\qin}) \leq J^{\rm{cs}}(\bf{\qin})$ necessarily, and the ratio is always non-negative. We first show the upper bound when considering one compressor; the ratios simplify to $R_{\min} = P_c(\qmin)/\qmin$ and $R_{\max} = P_c(\qmax)/\qmax$. Consider two thermal load profiles $\bf{q}^1$ and $\bf{q}^2$ of length $T \gg 0$, where $q^1(k) = \qmin$ for $T - \rm{round}(D \cdot \qmax) \leq k < T$ and $0$ elsewhere and $q^2(k) = \qmax$ for $0 \leq k <\rm{round}(D \cdot \qmin) $ and $0$ elsewhere. We choose $D$ and $T$ large enough that $\sum_k q^1(k)$ is approximately equal to $\sum_k q^2(k)$ and $\rm{round}(D \cdot \qmin) < T - \rm{round}(D \cdot \qmax)$. By our definitions, $\bf{q}^2$ is a load shifted version of $\bf{q}^1$, in that it satisfies $\sum^\tau_{k=0} q^2(k) \geq \sum^\tau_{k=0} q^1(k)$ for all $\tau \in [0, T]$. Thus without load shifting for $\bf{q}^1$, the power usage is approximately $P_c(\qmin) \times D \cdot \qmax$ and with load shifting to $\bf{q}^2$, the power usage is approximately $P_c(\qmax) \times D \cdot \qmin$. Thus the ratio of power use matches the upper bound $(R_{\max} - R_{\min})/R_{\min}$. It can be easily verified that this example attains the worst case ratio via convexity of $P_c$. Extension to multiple compressors follows an analogous argument. A similar construction can be assembled, where in $\bf{q}^2$, the most efficient compressor satisfies the required load at full capacity and in $\bf{q}^1$, the least efficient compressor satisfies the required load at the minimum capacity.
\end{proof}

Proposition \ref{prop:comp} provides a possible range of energy savings when utilizing load shifting. By the worst-case constructions of $\bf{\qin}$, we see that if the profile $\bf{\qin}$ fluctuates significantly, implementing load shifting can generate the most energy savings; however, if $\qin$ is relatively constant, the energy savings may be less. For the compressors operating in Butterball, the fractional difference is between $0 \leq \frac{J^{\rm{cs}}(\bf{\qin}) - J^*(\bf{\qin})}{J^*(\bf{\qin})} \leq 8.85$, suggesting significant energy savings through load shifting. We validate this in the next section.

\begin{figure*}[t!]
    \centering
    \includegraphics[width=500pt]{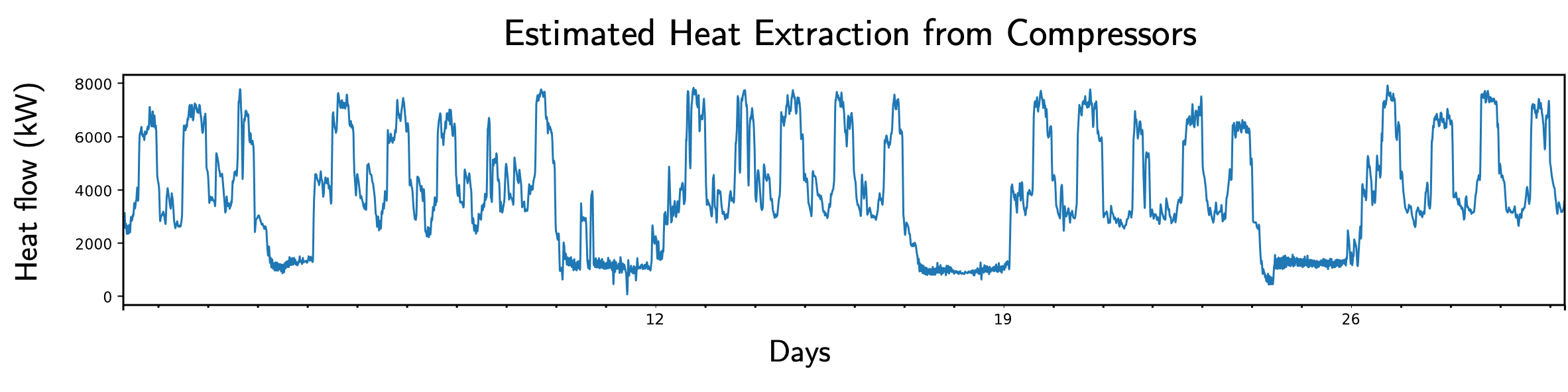}
    \caption{Over a month, we display the estimated refrigeration load serviced by the compressors, where each data point is associated with every minute in June 2023. In this figure, we also apply an average filter of 20 min. to smooth out the data. Data points for refrigeration load were calculated through estimates of compressor power and COP through sensor readings for each time point.}
    \label{fig:heat}
\end{figure*}

\section{The Butterball Facility: A Case Study on Opportunities in Compressor Sequencing}
\label{sec:case}

In this section, we evaluate the potential energy savings possible through compressor sequencing and load shifting on a case study of the Butterball facility.  Here, we use the predicted refrigeration load profile from the Butterball facility during the month of June 2023, depicted in Figure \ref{fig:heat}, to serve as a prototypical example for our case study. This profile was estimated using direct measurements of power usage of the compressors, as well derived coefficient of performance or COP (sometimes CP or CoP) of the refrigeration system, which is defined as the ratio of useful cooling to work done by the compressors. We observe that the thermal profile is relatively cyclic, where the peaks correspond to working hours during the work week (when new product requiring cooling typically enters during standard operating hours), and the lower plateaus represent weekends (representing off hours of the facility).

In the Butterball facility, all compressors in operation are screw compressors, which use a screw thread to trap and compress a volume of gas. For screw compressors, the main control parameter to modulate the capacity is a continuous slide valve control.\footnote{It is also sometimes possible to modulate the capacity of a screw compressor through speed control with a variable speed drive. While this paper does not explicitly consider this method, our analysis extends to this scenario as well.} Increasing the slide valve will expose more of the screw thread, increasing the volume of gas to be compressed and increases the capacity of the compressor. As highlighted previously, compressors operate at their highest efficiency when running at full capacity. The slide valve being at the minimum position corresponds to the minimum heat capacity $\qmin$ (and similarly for $\qmax$) for compressor $c$.
 
The main energy savings studied in this paper is summarized in Table \ref{table:main}. First, we see that there are significant energy savings if good fixed orders are used. When using only optimal compressor sequencing, the possible energy savings are up to only $5\%$ as compared to the best fixed order. However, when utilizing compressor sequencing with load shifting, we see energy savings for up to $20\%$. We also compare the performance of a simple, online version of a compressor sequencing algorithm with load shifting, given in Algorithm \ref{alg:online}. In this online algorithm, the compressors service either the mean refrigeration or the required refrigeration load at full capacity. We see that this simple implementation achieves similar energy savings to the optimal compressor sequencing with load shifting. Thus we validate load shifting as a viable mechanism for garnering energy savings with regards to compressor sequencing. In future work, we will extend these results to construct realizable control algorithms to offer actualized energy savings.

\begin{algorithm}
\caption{Online Load-Shifting Algorithm}\label{alg:online}
\begin{algorithmic}
\Require $\bf{\qin}$, $(k, \qk, \qtot) \gets (0, 0, 0)$, $\bf{q}_c \gets \bf{0}$ for all $c \in \C$
\For{$k \leq T$}
\State $\qk \gets \qk + \qin(k)$
\State $\qtot \gets \max\{\qk, \rm{mean}(\bf{\qin})\}$
\For{$c$ in $\O^{\rm{sh}}$}
\If{$\qtot > 0$}
\State $q_c \gets \qmax$
\State $\qtot \gets \qtot - \qmax$
\State $\qk \gets \qk - \qmax$
\EndIf
\EndFor
\State $k \gets k + 1$
\EndFor
\end{algorithmic}
\end{algorithm}

\begin{table}[h!]
\centering
\begin{tabular}{ | r | l |  }
\hline 
 Methodology & Average Power  \\
 \hline 
 Worst Fixed Order & $856.7$ kW  \\  
 Best Fixed Order & $562.3$ kW \\ 
 Compressor Sequencing (C.S) & $551.0$ kW \\ 
 Online C.S with Load Shifting &  $444.3$ kW\\
 C.S with Load Shifting & $443.5$ kW \\ 
 \hline 
\end{tabular}
\caption{Cost of Algorithms}
\label{table:main}
\end{table}

\section{Conclusion}
\label{sec:conc}

Improving the energy efficiency of industrial refrigeration systems has the potential to have significant impacts on energy usage. This work focuses on optimizing the efficiency of the compressors of the refrigeration cycle, as the compressors comprise the largest energy demand of the system. Specifically, we focus on the problem of compressor sequencing, where the operating conditions of the compressors are optimized given a thermal load requirement. We first study the compressor sequencing problem and characterize the structure of the optimal solution as a variant of a water-filling algorithm. However, deriving the optimal solution is computationally hard in general. Thus we utilize the mechanism of load shifting, where we utilize pre-cooling, to derive more efficient compressor sequencing algorithms. We observe that load shifting promotes straight-forward compressor sequencing algorithms and allows for additional energy savings. We validate our results on an operational Butterball facility. Through sensor data, we see that load shifting can potentially reduce energy use to up to $20\%$. In the future, we will also explore possible extensions, such as incorporating more realistic heat flow models, tackling load shifting through the lens of inventory control, exploiting dynamic energy prices to inform precooling, and derive stable implementations to be employed on facilities.

\bibliographystyle{ieeetr}
\bibliography{references.bib}
\end{document}